\documentclass[10pt,a4paper]{article}

\usepackage{amsmath,amssymb,amsfonts,amsthm}
\usepackage{textcomp}
\usepackage{t1enc}
\usepackage[english]{babel}
\usepackage[latin1]{inputenc}
\usepackage{color}
\usepackage{verbatim}
\usepackage{hyperref}
\newtheorem{theorem}{Theorem}[section]
\newtheorem{lemma}[theorem]{Lemma}

\theoremstyle{definition}

\usepackage{enumerate}

\definecolor{bronze}{rgb}{0.8, 0.5, 0.2}

\renewcommand{\t}[1]{\tilde{#1}}

\newcommand{\ra}{\rightarrow}

\newcommand{\be}{\begin{equation}}
\newcommand{\me}{\mathbf{E}}
\newcommand{\ee}{\end{equation}}
\newcommand{\mb}{\mathbf{B}}

\newcommand{\Rt}{\mathbb{R}^3}

\title{Extremal black hole initial data deformations}
\author{Andr\'es Ace\~na$^{1}$ and Mar\'\i a  E. Gabach Cl\'ement$^{2}$\\ 
\\  
  $^{1}$Facultad de Ciencias Exactas y Naturales,\\
  Universidad Nacional de Cuyo, \\ CONICET, Mendoza, Argentina. \\
  $^2$Facultad de Matem\'atica, Astronom\'{i}a y F\'{i}sica,\\
  Universidad Nacional de C\'ordoba,\\
  Instituto de F\'{\i}sica Enrique Gaviola, CONICET,\\
  C\'ordoba, Argentina.  \\
  }
\date{}

\begin{document}
\maketitle
\begin{abstract}
We study deformations of axially symmetric initial data for Einstein-Maxwell equations 
satisfying the time-rotation ($t$-$\phi$) symmetry and containing one asymptotically 
 cylindrical end and one asymptotically flat end. We find that the $t$-$\phi$ symmetry implies the 
 existence of a family of deformed data having the same horizon structure. This result allows us to
 measure how close solutions to Lichnerowicz equation are when arising from nearby free data.
 \end{abstract}

\section{Introduction}

The observation  that 
wormhole initial data (black hole data having two asymptotically flat ends) rapidly evolve to trumpet initial data (one 
asymptotically flat and one cylindrical 
end) (see \cite{Hannam:2006vv}, \cite{Hannam:2008sg} and references therein) motivated the use of trumpet data to study numerical binary collisions since, in this way, the gauge
evolution and the initial noise  in wave quantities would be minimized. This then inspired an extensive
study of
initial data for Einstein equations having cylindrical ends, both from the numerical relativity 
community \cite{Baumgarte:2008jc}, \cite{Hannam:2009ib}, \cite{Immerman:2009ns}, \cite{Hinder:2010vn}, \cite{Dietrich:2013rua}, \cite{Dennison:2014eta} and the mathematical 
relativity side \cite{Dain:2008yu}, \cite{gabach09}, \cite{Waxenegger:2011ci}, 
\cite{Schoen:2012nh}, \cite{Chrusciel:2012np}, \cite{Chrusciel:2012nv}. There seems to be a close relation between the presence of a cylindrical end and 
certain extremality condition  suggested in part, by the behavior of stationary solutions 
like Kerr-Newman and also by the fact \cite{gabach09} that given a mono-parametric family of conformally flat initial data
having a wormhole structure, with given angular momentum and charges, then there exists a singular 
limit as the parameter goes to zero, where the asymptotic structure changes to  trumpet-like and the angular momentum and 
charges are maximal for given mass. 
This reinforces the interest in studying initial data with cylindrical ends
in  an attempt to understand cosmic censorship issues, black hole formation, conical singularities appearing
in stationary multi-black hole solutions, etc.

Initial data with more than one cylindrical end, \textit{i.e.} representing many \textit{extremal} 
black holes, are specially important. They include data for the Majumdar-Papapetrou solution 
\cite{Weinstein:1994bn}, consisting of black holes of the extremal Reissner-Nordstr\"om type. It is 
the only static multi black hole solution of Einstein-Maxwell equations in equilibrium known to us. Moreover
one would expect it to be the unique electro-vacuum, stationary solution with disconnected horizon. Nevertheless, 
the proof of this result and a complete analysis of its stability are lacking.
Motivated by these open problems it is our aim here to understand perturbations of electromagnetic fields in initial data for Einstein-Maxwell
with cylindrical ends. To us, this is a first step in the study of deformations of the full 
4-dimensional Majumdar-Papapetrou solution.

In the past five years there has been increasing interest in developing the mathematical tools appropriate
to deal with the problem of initial data with cylindrical ends. In \cite{Dain:2008yu} weighted Sobolev spaces were
used to prove existence of an  extremal solution with one cylindrical end as a special limit of Bowen-York 
initial data. This was generalized in \cite{gabach09} to conformally flat initial data. Then, Waxenegger
\textit{et al} \cite{Waxenegger:2011ci} adapted the theorem of sub and
super solution on weighted H\"older spaces to prove the same result without invoking the singular extremal limit.
On the other hand, in \cite{Dain:2010uh}, deformations of extreme Kerr black holes were studied. It was proven that for compactly supported perturbations, there
exists a unique family of nearby initial data, that have the same horizon structure as extreme Kerr  
but greater ADM mass. For that result a specific property of 
extreme Kerr's metric was explicitly used in the proof.
In more general terms Chrusciel \textit{et al}  \cite{Chrusciel:2012np} have studied solutions to Lichnerowicz equation. They proved 
existence of vacuum 
initial data with positive scalar curvature, having a number of asymptotically flat and cylindrical ends.
This is an important existence result that extends previous ones by  Choquet-Bruhat \textit{et al} 
\cite{Choquet99} to manifolds with cylindrical (or periodic or hyperbolic) ends. Uniqueness of solution however,
has not been dealt with in \cite{Chrusciel:2012np} due mainly to the methods 
used there.

In this article we are interested in electro-vacuum initial data with an asymptotically flat end and one 
cylindrical end, representing the black hole horizon. We address the problem of how \textit{close} solutions to Lichnerowicz equation 
are, when they arise from \textit{close} free data.  
The idea is thus to consider two sets of free data for Lichnerowicz equation that are close in a 
certain norm, and analyze how close the corresponding initial data found from them are. We
choose the free data sets as one being a deformation of the other one. This extends the result of 
\cite{Dain:2010uh} to more general, axially symmetric initial data for Einstein-Maxwell
equations having the $t$-$\phi$ symmetry. 
Basically this extra symmetry gives a positivity condition (Yamabe positivity) that replaces the 
explicit property of Kerr used in \cite{Dain:2010uh}. Moreover, we also abandon 
the vacuum hypothesis, in view of our later study of the Majumdar-Papapetrou solution.

The manuscript is organized as follows: In section \ref{secmain} we present Einstein constraints and  describe 
in detail the hypotheses we work with, axial symmetry and time-rotation symmetry. We show how they lead to the Lichnerowicz equation and the
asymptotic boundary conditions. We present our main result, Theorem  \ref{thm} and discuss its scope and 
implications afterwards. In section \ref{secproof} we present the proof of Theorem \ref{thm}.

\section{Main result}\label{secmain}

Consider a 3-dimensional surface $M=\mathbb R^3 \setminus \{0\}$. An initial data for the Einstein-Maxwell equations is 
a set $(M,g_{ij},K_{ij}, E^i,B^i)$ where $g_{ij}$ is the 3-metric on $M$, $K_{ij}$ is the 
extrinsic curvature tensor and $E^i$, $B^i$ are the electromagnetic fields on $M$. 
This set of fields satisfies the constraints on $M$:
\be\label{ham}
R+K^2-K_{ij}K^{ij}=2(E_iE^i+B_iB^i)
\ee
\be\label{mom}
D_jK_i^j-D_iK=-2\epsilon_{ijk}E^jB^k
\ee
\be\label{maxcon}
D_iE^i=0,\qquad D_iB^i=0
\ee
where $K=K_{ij}g^{ij}$, $D_i$, $R$ and $\epsilon_{ijk}$ are respectively the covariant 
derivative, the curvature scalar and the volume form associated to the metric $g_{ij}$. 
For simplicity, we will not consider
electromagnetic currents in \eqref{mom}, \textit{i.e.} $j_i=-2\epsilon_{ijk}E^jB^k=0$. This is  
a technical assumption to make the equations and calculations easier, but could be removed 
without altering the basic results of this article.

We will focus on initial data satisfying the above equations and the following 
three hypotheses:

\vspace{0.5cm}

\textit{H1. Axial symmetry}. We consider axially symmetric initial data, namely, we assume that there exists a Killing vector field $\eta$ tangential to $M$ with complete closed 
orbits, such that $\mathcal L_\eta g_{ij}=0$, $\mathcal L_\eta K_{ij}=0$, $\mathcal L_\eta E^{i}=0$, $\mathcal L_\eta B^{i}=0$.  In cylindrical coordinates 
$(\rho,z,\phi)$ we write $\eta^i=(\partial_\phi)^i$ and axial symmetry implies in particular that the fields above will not depend on $\phi$. Moreover, we will see below that this 
assumption allows us to write the metric, extrinsic curvature and electromagnetic fields in a simple manner in terms of scalar potentials.

\vspace{0.5cm}

\textit{H2. Time-rotation symmetry.} Besides axial symmetry we impose a discrete symmetry, namely, time-rotation symmetry. In terms of the 
initial data and the coordinates associated with the axial symmetry, this means that under the map 
$\phi\to-\phi$ the initial data map as (see 
Appendix \ref{Aptimerotation})
\be\label{tphi1}
g_{ij}\to g_{ij}, \qquad K_{ij}\to-K_{ij} 
\ee
and 
\be\label{tphi2}
E^i\to E^i,\qquad B^i\to B^i.
\ee
Initial data satisfying this symmetry turn out to be maximal and has been called ``momentarily stationary'', as this symmetry is to a stationary 
space-time what time-symmetry is to a static space-time \cite{Bardeen70}, \cite{Hawking73b}, 
\cite{Brandt:1996si}. In the treatment below it 
will be highlighted why we need this symmetry in order for our equations to be written in a particular form, without being too restrictive as 
to forbid the consideration of dynamical space-times.

It is important to remark that the time-rotation symmetry implies (see 
\cite{Brandt:1996si}) maximality $K=0$ and moreover, due to the Hamiltonian constraint \eqref{ham}, also 
$R\geq0$. This in turn means that $(M,g_{ji})$ satisfies the positivity condition
\be\label{yamabe}
\int_M|\partial f|_{g}^2+ Rf^2 d\mu_g>0
\ee
for all $f\in C^\infty_0$, where $\partial$
denotes partial derivatives and the norm, curvature scalar and volume element 
$d\mu_g$ are taken with respect to $g$. For later purposes, we will say that $(M,g_{ji})$ 
satisfying \eqref{yamabe} is in the positive Yamabe class $\mathcal Y^+$.

\vspace{0.5cm}

\textit{H3. Asymptotic structure}. The manifold $M=\mathbb R^3\setminus\{0\}$ has an asymptotically flat end and we take the origin to be a cylindrical end. This means \cite{Chrusciel:2012np} 
that the cylindrical end is identified with the 
product $\mathbb R^+\times N$, where $N$ is compact and the asymptotic metric is conformal (with bounded conformal factor) to 
\be\label{metriccyl}
\hat g=dx^2+h+\mathcal O(e^{-a x})
\ee
for a metric $h$ on $N$ and some positive constant $a$. Moreover we will restrict our study to the topologically spherical case $N=S^2$, and take $h$ to be
the standard metric on the unit sphere. 

\vspace{0.5cm}
We will approach the constraint equations by using the Conformal Method, \cite{Choquet99}. Consider the rescaling 
\be\label{rescaling}
g_{ij}=\Phi^4\tilde g_{ij},\quad K_{ij}=\Phi^{-2}\tilde K_{ij},\quad,E^i=\Phi^{-6}\tilde E^i,\quad B^i=\Phi^{-6}\tilde B^i
\ee
where $\Phi>0$ and for simplicity, we take the conformal metric to be
\be\label{metric}
\tilde g_{ij}=e^{2q}(d\rho^2+dz^2)+\rho^2 d\phi^2
\ee
where $q$ does not depend on $\phi$. This rescaled conformal metric $\tilde g_{ij}$ is not the most general axially 
symmetric metric satisfying \eqref{tphi1} (see eq. (5) in \cite{Acena:2010ws}). Nevertheless, it is not difficult to see that
the same procedure can be made  for that
more general metric.

Under this rescaling, the constraints read
\be\label{hammax}
\tilde D_i\tilde D^i\Phi=\frac{1}{8}\tilde R\Phi-\frac{\tilde K_{ij}\tilde K^{ij}}{8\Phi^7}-\frac{\tilde E_i\tilde E^i+\tilde B_i\tilde B^i}{4\Phi^3}
\ee
\be\label{mommax}
\tilde D_j\tilde K_i^j=0,\quad \tilde D_i\tilde E^i=0,\quad \tilde D_i\tilde B^i=0.
\ee
Here $\tilde R$ is the curvature scalar associated to $\tilde g_{ij}$
 
In electro-vacuum and axial symmetry, the fact that $M$ is simply connected implies \cite{Chrusciel:2009ki} the 
existence of potentials $\omega$, $\psi$ and $\chi$ given by 
\be\label{defomega}
\tilde K^{ij}=\frac{2}{\rho^2}\tilde S^{(i}\eta^{j)},
\qquad \tilde S^i=\frac{1}{2\rho^2}\tilde \epsilon^{ijk}\eta_j\partial_k\omega, 
\ee
\be\label{defF}
\partial_i\chi=F_{ji}\eta^j,\qquad \partial_i\psi= *F_{ji}\eta^j,
\ee
such that the momentum and Maxwell constraints \eqref{mommax} are automatically satisfied (see \cite{Dain06c} for a proof in the momentum 
case and the appendix \ref{ApEinstein} 
for the electromagnetic case). Here $F_{ij}$ is the 4-dimensional electromagnetic tensor, which can be constructed in the standard way from $E^i$ and $B^i$ \eqref{FfromEB}.

The values of the potentials $\omega$, $\psi$ and $\chi$  are constant on each connected component of the symmetry axis $\Gamma:=\{\rho=0\}$ and 
give the angular momentum $J$, 
electric charge $Q_E$ and magnetic charge $Q_B$ respectively \cite{Chrusciel:2009ki}:
\be\label{bndcond}
J=\frac{\omega_--\omega_+}{8},\qquad Q_E=\frac{\psi_--\psi_+}{2},\qquad Q_B=\frac{\chi_--\chi_+}{2}, 
\ee
where we denote by $\omega_+:=\omega(\rho=0, z>0)$, $\omega_-:=\omega(\rho=0, z<0)$ the values of the function $\omega$ on the axis, 
at positive and negative values of $z$ respectively. Analogous notation holds for the electromagnetic potentials.

In terms of these potentials, the symmetry conditions \eqref{tphi1}-\eqref{tphi2} translate into the 
following expressions appearing in the Hamiltonian constraint
\be\label{fieldspot}
\tilde K_{ij}\tilde K^{ij}=e^{-2q}\frac{(\partial\omega)^2}{2\rho^4},\qquad 
\tilde E_i\tilde E^i=e^{-2q}\frac{(\partial\psi)^2}{\rho^2},\qquad \tilde B_i\tilde B^i=e^{-2q}\frac{(\partial\chi)^2}{\rho^2},
\ee
where the norms are taken with respect to the flat metric on $\mathbb{R}^3$.
In general, for data not satisfying the symmetry conditions,  a $\geq$ sign holds in the three equations in \eqref{fieldspot} (see Appendix \ref{Aptimerotation}).

The scalar curvature in terms of the metric function $q$ is given by 
\be\label{R}
\tilde R=-2e^{-2q}\Delta_2q
\ee
with 
$\Delta_2=\partial^2_\rho+\partial^2_z$.

With these variables, the only non-trivial equation left is the Hamiltonian 
constraint \eqref{ham}, which takes the form
\be\label{const}
\Delta\Phi=-\frac{\Delta_2q}{4}\Phi-\frac{(\partial\omega)^2}{16\rho^4\Phi^7}-\frac{(\partial\psi)^2+(\partial\chi)^2}{4\rho^2\Phi^3},
\ee
where $\Delta=\partial^2_\rho+\rho^{-1}\partial_\rho+\partial^2_z$. 

This equation, known as the Lichnerowicz equation, is a non-linear equation for the conformal factor $\Phi$. The set of functions
$\mathcal F:=(q,\omega, \psi,\chi)$ is 
known as \textit{free data}, and can be freely prescribed, made to satisfy the asymptotic conditions appropriate for the problem at hand. Once 
$\Phi$ is found, we can construct the initial data as follows. From the prescribed function $q$ we have the conformal metric $\tilde{g}_{ij}$ \eqref{metric}. With the obtained conformal factor we calculate the metric $g_{ij}$, and using the prescribed functions $\psi$ and $\chi$ we can calculate $E^i$ \eqref{Eipot} and $B^i$ \eqref{Bipot}. From $\omega$ we calculate $\t{K}_{ij}$ \eqref{defomega} and rescaling obtain $K_{ij}$. Therefore we finally have $g_{ij}$, $K_{ij}$, $E^i$ and $B^i$.

Next we investigate the conditions that the functions 
$\Phi$ and $q$ in the metric \eqref{metric} must satisfy at the cylindrical end. We write $g$ in 
spherical coordinates $(r,\theta,\phi)$ and make the change $x=-\ln r$,
\be
g=r^2\Phi^4[e^{2q}dx^2+e^{2q}d\theta^2+\sin^2\theta d\phi].
\ee
Thus, by comparison with \eqref{metriccyl} we obtain that the conditions for the data on the cylindrical end $r\to0$ are
\be\label{phicil}
\Phi=\mathcal O(r^{-1/2}),\qquad q=\mathcal O(1).
\ee

In virtue of equation \eqref{const} and the regularity near the symmetry axis (see 
\cite{Rinne:thesis}) we obtain conditions for the derivatives of the function $q$ and the potentials on the cylindrical 
end
\be\label{derivq}
\Delta_2q=\mathcal O(r^{-2})
\ee
\be\label{derivpot}
|\partial\omega|^2=\mathcal O (r^{-2}\sin^6\theta ), \quad |\partial\psi|^2=\mathcal O (r^{-2}\sin^2\theta ), \quad |\partial\chi|^2=\mathcal O (r^{-2}\sin^2\theta ). 
\ee

Finally, recall that the Yamabe condition \eqref{yamabe} is conformally 
invariant, and therefore if $\tilde g_{ij}$ is conformally related to $g_{ij}$, 
then $(M,\tilde g_{ij})$ also belongs to $\mathcal Y^+$, that is
\be\label{Yamcond}
\int_M|\partial f|_{\tilde g}^2+ \tilde Rf^2 d\mu_{\tilde g}>0.
\ee

\vspace{0.5cm}

The question we want to address is the following. Consider two sets of free data $\mathcal F_0$ and 
$\mathcal F$ giving rise, through \eqref{const}, to corresponding conformal 
factors $\Phi_0, \Phi$ and thus, to initial data satisfying the hypothesis \textit{H1-H3} above. If the free 
data are close in certain norm, how close are the data constructed from them? Clearly, this 
depends mainly on the relative size of the conformal factors. To study this problem 
we will think of $\mathcal F$ as a deformation of the set $\mathcal F_0$, then look for a solution to 
Lichnerowicz equation close to $\Phi_0$ and finally, estimate its relative size. 

Assume $(\Phi_0,q_0,\omega_0,\psi_0,\chi_0)$ satisfy \eqref{const}. Let $|\lambda|$ be a sufficiently small number, take
\be\label{deformation}
q_0\to q_0+\lambda q,\quad \omega_0\to\omega_0+\lambda\omega,\quad\psi_0\to\psi_0+\lambda\psi,\quad\chi_0\to\chi_0+\lambda\chi
\ee
for appropriate axially symmetric functions $q,\omega, \psi, \chi$ and write 
\be
\Phi_0\to \Phi:=\Phi_0+u.
\ee
We demand the perturbed function $\Phi=\Phi_0+u$ to satisfy Lichnerowicz equation \eqref{const} and write the resulting equation for $u$ as
\be\label{eqG}
G(\lambda,u)=0
\ee
with
\begin{eqnarray}\nonumber
G(\lambda,u)=\Delta u+\frac{\Delta_2q_0 u}{4}+\frac{\lambda}{4}\Delta_2q (\Phi_0+u)+\frac{(\partial\omega_0+\lambda\partial\omega)^2}{16\rho^4(\Phi_0+u)^7}-\frac{(\partial\omega_0)^2}{16\rho^4\Phi_0^7}+\\ \label{G}+
\frac{(\partial\psi_0+\lambda\partial\psi)^2}{4\rho^2(\Phi_0+u)^3}-\frac{(\partial\psi_0)^2}{4\rho^2\Phi_0^3}+
\frac{(\partial\chi_0+\lambda\partial\chi)^2}{4\rho^2(\Phi_0+u)^3}-\frac{(\partial\chi_0)^2}{4\rho^2\Phi_0^3}.
\end{eqnarray}
Clearly, if $\lambda=0$, we recover the Lichnerowicz equation for the background solution 
$\Phi_0$ in the form 
\be
G(0,0)=0.
\ee

Our main result, presented in the next theorem proves that there exists a unique solution $u$ of \eqref{eqG} close to the background $(0,0)$ for each small enough $\lambda$.

The weighted Lebesgue spaces $L'^2_\delta$ \cite{Bartnik86}, with weight $\delta\in \mathbb R$ are the spaces 
of measurable functions in $L^2_{loc}(\mathbb R^3\setminus\{0\})$ such that the norms
\begin{equation}
 \|u\|_{L'^2_{\delta}}=\left[\int_{\mathbb R^3\setminus\{0\}}|u|^2r^{-2\delta -3}\right]^{1/2}
\end{equation}
are finite. As usual the weighted Sobolev spaces $H'^k_\delta$ are defined with norms
\begin{equation}
 \|u\|_{H'^k_\delta}=\sum_{j=0}^k\|D^ju\|_{L'^2_{\delta-j}}.
\end{equation}

\begin{theorem}\label{thm}
Let $q, \omega, \psi, \chi\in C^\infty_0(\mathbb R^3\setminus\Gamma)$ be arbitrary smooth axially symmetric functions. Then,
there is $\lambda_0>0$ such that for all $\lambda\in (-\lambda_0, \lambda_0)$  there
exists a solution $u(\lambda)\in H^{'2}_{-1/2}$ of equation
\eqref{eqG}. The solution $u(\lambda)$ is continuously differentiable in
$\lambda$ and satisfies $\Phi_0+u(\lambda)>0$. Moreover, for small $\lambda$
and small $u$ (in the norm $ H^{'2}_{-1/2}$) the solution $u(\lambda)$ is the
unique solution of equation \eqref{eqG}.   
\end{theorem}

Once $u$ is found, we can re-construct the perturbed initial data in the same way as it was explained above.

With this construction, a different wording of Theorem \ref{thm} can be presented as follows:  given 
initial data
$\mathcal D^0:=(M,g^0_{ij}, K^0_{ij}, E^{0i}, B^{0i})$ that satisfy hypotheses \textit{H1-H3}, 
with angular 
momentum $J$ and electromagnetic charges $Q_E$ and $Q_B$, there exists a mono-parametric family of initial 
data sets $(M, g_{ij}(\lambda),  K_{ij}(\lambda), E^{i}(\lambda),B^{i}(\lambda) )$, unique for each 
$\lambda$ sufficiently close to zero such that 
\begin{enumerate}[(i)]
\item $g_{ij}(0)= g^0_{ij}$,  $K_{ij}(0)=K^0_{ij}$, $E^i(0)=E^{0i}$, $B^i(0)=B^{0i}$. The family is
  differentiable in $\lambda$ and it is close to $\mathcal D^0$ with respect to
  an appropriate norm which involves two derivatives.

\item The data has the same asymptotic geometry as $\mathcal D^0$. The angular momentum, charges  and the area of the cylindrical end in the
  family do not depend on $\lambda$, they have the same value as in $\mathcal D^0$.

\item The  data are axially symmetric and time-rotational symmetric. 

\end{enumerate}
The $\lambda$-dependent initial data is to be constructed from the given functions \eqref{deformation} and the solution $u(\lambda)$ to equation \eqref{eqG}. That is, one must solve equations \eqref{defomega}-\eqref{defF}  for $\tilde K_{ij} (\lambda)$ and $F_{ij}(\lambda)$ and use \eqref{metric} to obtain the metric in terms of $\lambda$.

Before going to the proof of Theorem \ref{thm}, we want to make some remarks.

\begin{itemize}

\item Several known black hole solutions fit into the hypotheses described above. In particular the extreme Bowen-York initial data built in 
\cite{Dain:2008yu}, the $\{t=0\}$ slice in extreme 
Kerr and extreme Reissner-Nordstr\"om black holes. As we explained in the introduction, the results of \cite{Dain:2010uh} are of course included in 
Theorem \ref{thm} for the vacuum, extreme Kerr case. Moreover, a $\{t=0\}$ slice of the axially symmetric 
Majumdar-Papapetrou solution also satisfies the hypotheses, but contains many cylindrical ends (as many as black holes are described). This case is of particular
interest for us and will be dealt with in a subsequent paper. The main difficulty that the many ends 
bring into the problem is the appropriate choice and treatment of the Sobolev spaces involved. 

\item The method of proof we use not only gives us existence but also \textit{uniqueness} of solution for 
each $\lambda$. Moreover we also obtain an estimate on the 
perturbed conformal factor in terms of the background one $0\leq\sqrt{r}\Phi\leq \max(\sqrt{r}\Phi_0)+C_1\sqrt{C_2}/2$  where $C_1,C_2$ are 
constants depending on $\Phi_0$ (see eq's. \eqref{posP2}, \eqref{defXi}). In turn, this estimation on the 
conformal factor and the size of $\lambda$ allows us to control how different the initial data 
$(M,g_{ij},K_{ij}, E^i,B^i)$ and $(M, g_{ij}^0,K_{ij}^0, E^{0i},B^{0i})$ are.

\item Axial symmetry is required to define in a well manner the 
angular momentum of the initial data. Time-rotation symmetry is used to simplify the analysis of the constraint equations as it gives an explicit and simple
 relation
between the fields $K_{ij}, E^i, B^i$ and the potentials. We believe the most important ingredient for our purposes that we obtain from this symmetry is the Yamabe positivity. 

\item In \cite{Dain:2010uh} it was shown that the weighted Sobolev space $H^{'2}_{-1/2}$ 
is specially appropriate for the study of small perturbations of 
solutions to Lichnerowicz equation with a cylindrical end. By small here we mean that the 
structure of the cylindrical end is unchanged by the perturbation (see \cite{Dain:2010uh} for details). This is due to the fact that
the background function satisfies $\Phi_0=\mathcal O(r^{-1/2})$ asymptotically at the cylindrical end. The perturbation is not meant to change the asymptotic structure
of the end, which translates to $u=o(r^{-1/2})$ at the end. This behavior is captured by the $-1/2$ weight in the Sobolev space.

\item The compact support away from the symmetry axis of the metric function  $q$ is required by the regularity desired on the metric $\tilde g$, this guarantees
that there will not be a conical singularity on the axis. On the other hand, in view of \eqref{bndcond}, the compact support of $\omega, \psi, \chi$
implies that there is no change in the angular momentum and charges
of the data. Moreover, the whole horizon structure remains unaltered, in
particular the horizon area will be the same as in the background. This can be seen as follows. The horizon area is computed as
\be
A=\lim_{r\to0}\int_{B_r}ds_{g}
\ee
where $B_r$ is a coordinate ball of radius $r$ and $ds_g$ is the area element with respect to the metric $g_{ij}$. This integral can be written as
\be
A(\lambda)=\lim_{r\to0}\int_{B_r}\Phi^4r^2e^{q_0+\lambda q}\sin\theta d\theta d\phi=\lim_{r\to0}\int_{B_r}(\Phi_0+u)^4r^2e^{q_0+\lambda q}\sin\theta d\theta d\phi
\ee
and using the boundary conditions \eqref{phicil} and $u=o(r^{-1/2})$ we find $A=A_0$.

If one wants to alter, say, the angular momentum, then $\omega$ must have a precise asymptotic 
behavior at $r\to0$ and the axis $\rho=0$. We expect that a different treatment will be necessary
to deal with this case as it is likely that the perturbed solution $\Phi_0+u$ will no longer have the 
same asymptotic behavior, resulting probably in a different character for the end (changing from asymptotically 
cylindrical to asymptotically flat  or giving rise to a naked singularity).

\item The condition of positive Yamabe for the background data $(M,g_{ij})$, \eqref{yamabe}
does not imply a non-negative conformal scalar curvature, $\tilde R_0\geq0$ as is assumed 
in  \cite{Chrusciel:2012np}. That is, $\tilde R_0$ can attain positive, negative and zero values, we
only know that inequality \eqref{Yamcond} is satisfied. However, if on top of Yamabe positivity, we 
assume $\tilde R_0\geq0$, we can estimate how small the deformation parameter $\lambda$ needs 
to be  in order to guarantee the existence of a new solution. We find
\be
\lambda\leq-\frac{\Delta_2q_0}{\Delta_2q}.
\ee 
This condition arises from equation \eqref{R} and the results in \cite{Chrusciel:2012np}.
Note that this does not depend on the size of the $\omega, \psi,\chi$ functions, but only on the perturbation function $q$.

\end{itemize}

\section{Proof of main result}\label{secproof}
\subsection{Sketch of the proof}
The proof uses the Implicit Function theorem to show that there exists a unique solution to \eqref{eqG}. We first investigate the appropriate functional spaces where we expect 
to find the solution. Then prove that the operator $G$ is well defined and continuously differentiable on these spaces. Finally, we prove that the associated linear operator 
$DG$ is an isomorphism. In this last step we use the Riesz Representation Theorem to find a weak solution and then, a regularity theorem to prove that
the weak solution is a strong solution. The Yamabe condition 
\eqref{Yamcond} plays a key role in the last parts of the proof, as it serves as the 
coercivity condition needed for the application of  Riesz Theorem.

In this section we use several constants whose exact value is not relevant, we denote them by $C_i$.

\subsection{Sobolev spaces and neighborhoods used}

We will work with the only non-trivial constraint equation written as \eqref{eqG} and  look for
a solution $u$. In \cite{Dain:2010uh} the authors deal with an analogous map $G$, and choose $G:\mathbb{R}\times H'^2_{-1/2}\ra L'^2_{-5/2}$ considering the fall off behavior of the functions involved. As in our case the asymptotic behavior is the same, we choose the same function spaces.

We are considering the map $G:\mathbb{R}\times H'^2_{-1/2}\ra L'^2_{-5/2}$, but for a general
$u\in H'^2_{-1/2}$ the function $\Phi=\Phi_0+u$ does not have a definite sign. In order for $\Phi$ to
be a proper conformal factor we need it to be positive. As we take $\Phi_0>0$, then we can conjecture 
that if $u$ is small enough, then $\Phi$ is also going to be positive. There are some subtleties in the problem at hand, as we have a particular behavior at the cylindrical end. Even so, it is possible to prove the conjecture, that is, to show that that there is a neighborhood $V$ of $0$ in $H'^2_{-1/2}$ such that
\be
 \Phi_0+u>0.
\ee
We start by noting that as $\Phi_0$ is a proper conformal factor, then it is positive and bounded away from zero if we remove a neighborhood of the cylindrical end. Approaching the cylindrical end, $\Phi_0\ra r^{-1/2}$ as $r\ra 0$, thus we can conclude that there are positive constants $C_1$, $C_2$, $C_3$ and $C_4$ such that
\begin{equation}\label{des1}
 C_1\sqrt{r+C_2}\leq\sqrt{r}\Phi_0\leq C_3\sqrt{r+C_4}.
\end{equation}
The argument in \cite{Dain:2010uh} carries through. Consider the open ball of radius $\xi$ around the origin in $H'^2_{-1/2}$,
\begin{equation}
 V = \{v\in H'^2_{-1/2} : ||v||_{H'^2_{-1/2}}<\xi\},
\end{equation}
where $\xi>0$ is yet to be defined. From Lemma A.1 in \cite{Dain:2010uh} we have that if $\|u\|_{H'^2_{-1/2}}<\xi$ then there is a constant $C$ such that
\begin{equation}\label{posP2}
 \sqrt{r}|u|\leq C\xi.
\end{equation}
Given $\Phi_0$ satisfying \eqref{des1} we find
\begin{equation}\label{des2}
\sqrt{r}(\Phi_0+u) \geq C_1\sqrt{C_2} - C\xi =: C_5,
\end{equation}
and if we choose $\xi$ such that
\begin{equation}\label{defXi}
 0<\xi<\frac{C_1\sqrt{C_2}}{2C}
\end{equation}
then
\begin{equation}
 C_5> \frac{C_1\sqrt{C_2}}{2} >0,
\end{equation}
and therefore
\begin{equation}
 \Phi_0+u > 0.
\end{equation}
From now on $\xi$ and $V$ are fixed. The factor $1/2$ in the r.h.s. of \eqref{defXi} is a technical requirement needed later to perform some bounds.

\subsection{The map $G:\mathbb{R}\times V\ra L'^2_{-5/2}$ is well defined}
To prove that $G:\mathbb{R}\times V\ra L'^2_{-5/2}$ is a well defined map we evaluate $\|G\|_{L'^2_{-5/2}}$,
\begin{eqnarray}\label{GL1}
 \|G(\lambda,u)\|_{L'^2_{-5/2}} & \leq &  \|\Delta u\|_{L'^2_{-5/2}}+\\ \label{GL2}
& +&\Bigg\|\frac{\lambda\Delta_2q}{4}(\Phi_0+u)\Bigg\|_{L'^2_{-5/2}}+\Bigg\|\frac{\lambda\partial\omega(2\partial\omega_0+\lambda\partial\omega)}{16\rho^4(\Phi_0+u)^7}\Bigg\|_{L'^2_{-5/2}} 
 + \\ \label{GL3}
& +& \Bigg\|\frac{\lambda\partial\psi(2\partial\psi_0+\lambda\partial\psi)}{4\rho^2(\Phi_0+u)^3}\Bigg\|_{L'^2_{-5/2}}+
 \Bigg\|\frac{\lambda\partial\chi(2\partial\chi_0+\lambda\partial\chi)}{4\rho^2(\Phi_0+u)^3}\Bigg\|_{L'^2_{-5/2}}+\\ \label{GL4}
 & +& \Bigg\|\frac{\Delta_2q_0}{4}u\Bigg\|_{L'^2_{-5/2}} +\Bigg\|\frac{(\partial\omega_0)^2}{16\rho^4}\left[\frac{1}{(\Phi_0+u)^7}-\frac{1}{\Phi_0^7}\right]\Bigg\|_{L'^2_{-5/2}}+
 \\ \label{GL5}
  &+& \Bigg\|\frac{(\partial\psi_0)^2}{4\rho^2}\left[\frac{1}{(\Phi_0+u)^3}-\frac{1}{\Phi_0^3}\right]\Bigg\|_{L'^2_{-5/2}}+   \Bigg\|\frac{(\partial\chi_0)^2}{4\rho^2}\left[\frac{1}{(\Phi_0+u)^3}-\frac{1}{\Phi_0^3}\right]\Bigg\|_{L'^2_{-5/2}}.
\end{eqnarray}
The term on the r.h.s. of \eqref{GL1} is bounded by the definition of the $H'^2_{-1/2}$ norm. The terms in \eqref{GL2} and 
\eqref{GL3} are bounded due to the compact support of $q$, $\omega$, $\psi$ and $\chi$ respectively. The first term in \eqref{GL4}
is bounded due to $u\in H'^2_{-1/2}$ and the behavior of $q_0$ given in \eqref{derivq}. The remaining three norms are 
bounded due to the asymptotic conditions on the background functions \eqref{derivpot} 
together with the inequalities \eqref{des1} and \eqref{des2}. This can be seen as follows. Use the identity
\be
\frac{1}{a^p}-\frac{1}{b^p}=(b-a)\sum_{i=0}^{p-1}a^{i-p}b^{-1-i}
\ee
to write
\begin{equation}
\frac{1}{\Phi_0^p}- \frac{1}{(\Phi_0+u)^p} = r^{(p+1)/2} u H,
\end{equation}
where
\be
H=\sum_{i=0}^{p-1}[\sqrt{r}(\Phi_0+u)]^{i-p}[\sqrt{r}\Phi_0]^{-1-i}.
\ee
Using \eqref{des1} and \eqref{des2} we see that
\begin{equation}
 H\leq C_6,
\end{equation}
where $C_6$ is a constant that only depends on previous constants. Using the conditions \eqref{derivpot} we can bound for instance
\begin{eqnarray}
 && \Bigg\|\frac{(\partial\psi_0)^2}{4\rho^2}\left[\frac{1}{(\Phi_0+u)^3}-\frac{1}{\Phi_0^3}\right]\Bigg\|_{L'^2_{-5/2}} \leq 
 C_7\Bigg\|\frac{r^{-2}\sin^2\theta}{4\rho^2} (r^2 u H)\Bigg\|_{L'^2_{-5/2}} \\
 && = \frac{C_7 C_6}{4}\Bigg\|\frac{u}{r^2}\Bigg\|_{L'^2_{-5/2}} \leq \frac{C_7 C_6}{4}\|u\|_{L'^2_{-1/2}} \leq \frac{C_7 C_6}{4}\|u\|_{H'^2_{-1/2}}.
\end{eqnarray}
Applying the same argument to the other terms involving $\omega_0$ and $\chi_0$ completes the proof that 
$||G(\lambda,u)||_{L'^2_{-5/2}}$ is bounded and therefore the map is well-defined.

\subsection{The map G is continuously differentiable}
We now prove that $G$ is differentiable. To propose candidates for $D_1G(\lambda,u)$ and $D_2G(\lambda,u)$ we calculate the directional derivatives
\begin{eqnarray}
 D_1G(\lambda,u)[\gamma] := \frac{d}{dt}G(\lambda+t\gamma,u)\Bigg|_{t=0},\\
 D_2G(\lambda,u)[v] := \frac{d}{dt}G(\lambda,u+tv)\Bigg|_{t=0},
\end{eqnarray}
obtaining
\begin{eqnarray}\nonumber
 D_1G(\lambda,u)[\gamma] & = &\Bigg[\frac{\partial\omega(\partial\omega_0+\lambda\partial\omega)}{8\rho^4(\Phi_0+u)^7}+\frac{\Delta_2q}{4}(\Phi_0+u)\\\label{D1G}
 && +\frac{\partial\psi(\partial\psi_0+\lambda\partial\psi)}{2\rho^2(\Phi_0+u)^3}+\frac{\partial\chi(\partial\chi_0+\lambda\partial\chi)}{2\rho^2(\Phi_0+u)^3}\Bigg]\gamma, \\\nonumber
 D_2G(\lambda,u)[v] & = & \Delta v+\Bigg[-\frac{7(\partial\omega_0+\lambda\partial\omega)^2}{16\rho^4(\Phi_0+u)^8}+\frac{\Delta_2q_0+\lambda\Delta_2q}{4} \\\label{D2G}
 && -\frac{3(\partial\psi_0+\lambda\partial\psi)^2}{4\rho^2(\Phi_0+u)^4}-\frac{3(\partial\chi_0+\lambda\partial\chi)^2}{4\rho^2(\Phi_0+u)^4}\Bigg]v.
\end{eqnarray}
We show that the operators are bounded. The $L'^2_{-5/2}$ norm of each term inside square brackets in \eqref{D1G} is bounded due to compact support, the conditions \eqref{derivpot} and the inequality \eqref{des2}, then
\be
 \|D_1G(\lambda,u)[\gamma]\|_{L'^2_{-5/2}} \leq C_8 |\gamma|.
\ee
For the second operator the proof is a bit more tricky. We have
\begin{eqnarray}\nonumber
 \|D_2G(\lambda,u)[v]\|_{L'^2_{-5/2}} & \leq &
\|\Delta v\|_{L'^2_{-5/2}}+\Bigg\|\Bigg[-\frac{7\lambda\partial\omega(2\partial\omega_0+\lambda\partial\omega)}{16\rho^4(\Phi_0+u)^8}-
\frac{7(\partial\omega_0)^2}{16\rho^4(\Phi_0+u)^8}+\\
\nonumber& +&\frac{\Delta_2(q_0+\lambda q)}{4} -\frac{3\lambda\partial\psi(2\partial\psi_0+\lambda\partial\psi)}{4\rho^2(\Phi_0+u)^4} -\frac{3(\partial\psi_0)^2}{4\rho^2(\Phi_0+u)^4}-\\
&-&\frac{3\lambda\partial\chi(2\partial\chi_0+\lambda\partial\chi)}{4\rho^2(\Phi_0+u)^4} -\frac{3(\partial\chi_0)^2}{4\rho^2(\Phi_0+u)^4}
\Bigg]v\Bigg\|_{L'^2_{-5/2}} \\
\nonumber& = & \|\Delta v\|_{L'^2_{-5/2}}+\Bigg\|\Bigg[-\frac{7r^2\lambda\partial\omega(2\partial\omega_0+\lambda\partial\omega)}{16\rho^4(\Phi_0+u)^8}-
 \frac{7r^2(\partial\omega_0)^2}{16\rho^4(\Phi_0+u)^8}+\\
\nonumber&& +\frac{r^2\Delta_2(q_0+\lambda q)}{4} -\frac{3r^2\lambda\partial\psi(2\partial\psi_0+\lambda\partial\psi)}{4\rho^2(\Phi_0+u)^4} -\frac{3r^2(\partial\psi_0)^2}{4\rho^2(\Phi_0+u)^4}-\\
&&-\frac{3r^2\lambda\partial\chi(2\partial\chi_0+\lambda\partial\chi)}{4\rho^2(\Phi_0+u)^4} -\frac{3r^2(\partial\chi_0)^2}{4\rho^2(\Phi_0+u)^4}
\Bigg]\frac{v}{r^2}\Bigg\|_{L'^2_{-5/2}} \\
 & \leq &  \|v\|_{H'^2_{-1/2}} +C_9  \|v\|_{L'^2_{-1/2}} \leq C_{10}\|v\|_{H'^2_{-5/2}},
\end{eqnarray}
where again we have used that $v\in H'^2_{-1/2}$, equations \eqref{des2}, \eqref{derivq} and \eqref{derivpot} and the compact support of $\omega$, $q$, $\psi$ and $\chi$. This proves that the operators $D_1G$ and $D_2G$ are bounded.

To show that $D_1G$ is the partial Fr\'echet derivative we calculate
\begin{eqnarray}
 && G(\lambda+\gamma,u)-G(\lambda,u)-D_1G(\lambda,u)[\gamma]\\
 && =\Bigg[\frac{(\partial\omega)^2}{16\rho^4(\Phi_0+u)^7}+\frac{(\partial\psi)^2}{4\rho^2(\Phi_0+u)^3}+\frac{(\partial\chi)^2}{4\rho^2(\Phi_0+u)^3}\Bigg]\gamma^2,
\end{eqnarray}
and as $\omega$, $\psi$ and $\chi$ have compact support
\begin{equation}
 \|G(\lambda+\gamma,u)-G(\lambda,u)-D_1G(\lambda,u)[\gamma]\|_{L'^2_{-5/2}} \leq C_{11} |\gamma|^2,
\end{equation}
which shows that
\begin{equation}
 \lim_{\gamma\ra0}\frac{\|G(\lambda+\gamma,u)-G(\lambda,u)-D_1G(\lambda,u)[\gamma]\|_{L'^2_{-5/2}}}{|\gamma|}=0.
\end{equation}
For $D_2G$ we have
\begin{eqnarray}
 && G(\lambda,u+v)-G(\lambda,u)-D_2G(\lambda,u)[v] \\
 && = \frac{(\partial\omega_0+\lambda\partial\omega)^2}{16\rho^4}\Bigg[\frac{1}{(\Phi_0+u+v)^7}-\frac{1}{(\Phi_0+u)^7}+\frac{7v}{(\Phi_0+u)^8}\Bigg] \\
 && + \frac{(\partial\psi_0+\lambda\partial\psi)^2}{4\rho^2}\Bigg[\frac{1}{(\Phi_0+u+v)^3}-\frac{1}{(\Phi_0+u)^3}+\frac{3v}{(\Phi_0+u)^4}\Bigg] \\
 && + \frac{(\partial\chi_0+\lambda\partial\chi)^2}{4\rho^2}\Bigg[\frac{1}{(\Phi_0+u+v)^3}-\frac{1}{(\Phi_0+u)^3}+\frac{3v}{(\Phi_0+u)^4}\Bigg] \\
 && = \frac{(\partial\omega_0+\lambda\partial\omega)^2}{16\rho^4}r^\frac{9}{2}v^2H_1 + \left(\frac{(\partial\psi_0+\lambda\partial\psi)^2}{4\rho^2}+ \frac{(\partial\chi_0+\lambda\partial\chi)^2}{4\rho^2}\right)r^\frac{5}{2}v^2H_2 \\
 && = \Bigg[\frac{\lambda\partial\omega(2\partial\omega_0+\lambda\partial\omega)}{16\rho^4}r^6H_1+\frac{(\partial\omega_0)^2}{16\rho^4}r^6H_1+\\
 &&+\frac{\lambda\partial\psi(2\partial\psi_0+\lambda\partial\psi)}{4\rho^2}r^4 H_2 +\frac{(\partial\psi_0)^2}{4\rho^2}r^4 H_2\\
 &&+\frac{\lambda\partial\chi(2\partial\chi_0+\lambda\partial\chi)}{4\rho^2}r^4 H_2 +\frac{(\partial\chi_0)^2}{4\rho^2}r^4 H_2\Bigg]\frac{v^2}{r^\frac{3}{2}}
\end{eqnarray}
where $H_1$ is as in \cite{Dain:2010uh} and satisfies $|H_1|<C_{12}$ and $H_2$ is given by
\begin{equation}
 H_2=\frac{1}{[\sqrt{r}(\Phi_0+u+v)]^3[\sqrt{r}(\Phi_0+u)]^4}\sum_{i=0}^2 C_{i}[\sqrt{r}(\Phi_0+u)]^{2-i}(\sqrt{r}v)^i,
\end{equation}
with $C_{i}$ numerical constants and satisfy
\begin{equation}
 |H_2|\leq \frac{1}{(C_1\sqrt{r+C_2}-2C\xi)^7}\sum_{i=0}^2 |C_{i}|(C_3\sqrt{r+C_4}+C\xi)^{2-i}(C\xi)^i \leq C_{13}.
\end{equation}
Using that $\omega$, $\psi$ and $\chi$ have compact support
\be
\|G(\lambda,u+v)-G(\lambda,u)-D_2G(\lambda,u)[v]\|_{L'^2_{-5/2}}\leq C_{14}\Bigg\|\frac{v^2}{r^\frac{3}{2}}\Bigg\|_{L'^2_{-5/2}}\leq C_{15}\|v\|_{H'^2_{-1/2}}^2,
\ee
where the last inequality has been calculated in \cite{Dain:2010uh}. This proves that $D_2G$ is the Fr\'echet partial derivative.

The next step is to prove continuity of the derivatives. We compute \footnote{Note that eq (55) in \cite{Dain:2010uh} has a typo. It should be a $D_2$ derivative}.
\begin{eqnarray}
&& \|D_1G(\lambda_1,u)[\gamma]-D_1G(\lambda_2,u)[\gamma]\|_{L'^2_{-5/2}} \\
&& = \Bigg\|\Bigg[\frac{(\partial\omega)^2}{8\rho^4(\Phi_0+u)^7}+\frac{(\partial\psi)^2}{2\rho^2(\Phi_0+u)^3}+\frac{(\partial\chi)^2}{2\rho^2(\Phi_0+u)^3}\Bigg]\gamma(\lambda_1-\lambda_2)\Bigg\|_{L'^2_{-5/2}} \\
&& \leq \Bigg[\Bigg\|\frac{(\partial\omega)^2}{8\rho^4(\Phi_0+u)^7}\Bigg\|_{L'^2_{-5/2}}+\Bigg\|\frac{(\partial\psi)^2}{2\rho^2(\Phi_0+u)^3}+\Bigg\|_{L'^2_{-5/2}}\\
&& +\Bigg\|\frac{(\partial\chi)^2}{2\rho^2(\Phi_0+u)^3}+\Bigg\|_{L'^2_{-5/2}}\Bigg]|\gamma|\,|\lambda_1-\lambda_2| \\
&& \leq C_{16} |\gamma|\,|\lambda_1-\lambda_2|,
\end{eqnarray}
where again we used compact support and the bound \eqref{des2}.
We also compute
\begin{eqnarray}\nonumber
&& \|D_2G(\lambda,u_1)[v]-D_2G(\lambda,u_2)[v]\|_{L'^2_{-5/2}} \\
\nonumber&& = \Bigg\|\Bigg\{\frac{7(\partial\omega_0+\lambda\partial\omega)^2}{16\rho^4}\Bigg[\frac{1}{(\Phi_0+u_2)^8}- \frac{1}{(\Phi_0+u_1)^8}\Bigg] \\
\nonumber&& +\frac{3(\partial\psi_0+\lambda\partial\psi)^2}{4\rho^2}\Bigg[\frac{1}{(\Phi_0+u_2)^4}- \frac{1}{(\Phi_0+u_1)^4}\Bigg]+\\
\label{eqcon1}&&+\frac{3(\partial\chi_0+\lambda\partial\chi)^2}{4\rho^2}\Bigg[\frac{1}{(\Phi_0+u_2)^4}- \frac{1}{(\Phi_0+u_1)^4}\Bigg]\Bigg\}v\Bigg\|_{L'^2_{-5/2}} \\
\nonumber&& = \Bigg\|\Bigg\{\frac{7(\partial\omega_0+\lambda\partial\omega)^2}{16\rho^4}r^{9/2}H_3 +\frac{3(\partial\psi_0+\lambda\partial\psi)^2}{4\rho^2}r^{5/2}H_4 \\
\label{eqcon2}&&+\frac{3(\partial\chi_0+\lambda\partial\chi)^2}{4\rho^2}r^{5/2}H_4\Bigg\}v(u_2-u_1)\Bigg\|_{L'^2_{-5/2}} \\
\nonumber&& = \Bigg\|\Bigg\{\frac{7\lambda\partial\omega(2\partial\omega_0+\lambda\partial\omega)}{16\rho^4}r^6H_3+\frac{7(\partial\omega_0)^2}{16\rho^4}r^6H_3 +\\
\nonumber&&\frac{3\lambda\partial\psi(2\partial\psi_0+\lambda\partial\psi)}{4\rho^2}r^4H_4 +\frac{3(\partial\psi_0)^2}{4\rho^2}r^4 H_4 +\\
\label{eqcon3}&&\frac{3\lambda\partial\chi(2\partial\chi_0+\lambda\partial\chi)}{4\rho^2}r^4H_4+\frac{3(\partial\chi_0)^2}{4\rho^2}r^4 H_4 \Bigg\}\frac{v(u_2-u_1)}{r^{3/2}}\Bigg\|_{L'^2_{-5/2}} \\
&& \leq C_{17}  \Bigg\|\frac{v(u_2-u_1)}{r^\frac{3}{2}}\Bigg\|_{L'^2_{-5/2}} \leq C_{18} \|v\|_{H'^2_{-1/2}}\|u_1-u_2\|_{H'^2_{-1/2}},\label{eqcon5}
\end{eqnarray}
where to go from \eqref{eqcon1} to \eqref{eqcon2} we used
\be
r^{-9/2}\left(\frac{1}{(\Phi_0+u_1)^8}-\frac{1}{(\Phi_0+u_2)^8}\right)=(u_2-u_1)H_3
\ee
and
\be
r^{-5/2}\left(\frac{1}{(\Phi_0+u_1)^4}-\frac{1}{(\Phi_0+u_2)^4}\right)=(u_2-u_1)H_4
\ee
with
\be
H_3:=\sum_{i=0}^{7}(\sqrt{r}(\Phi_0+u_1))^{i-8}(\sqrt{r}(\Phi_0+u_2))^{-1-i},
\ee
\be
 H_4:=\sum_{i=0}^{3}(\sqrt{r}(\Phi_0+u_1))^{i-4}(\sqrt{r}(\Phi_0+u_2))^{-1-i}.
\ee
Lines \eqref{eqcon3} are merely a convenient re-writing of \eqref{eqcon2}. To go from \eqref{eqcon3} to \eqref{eqcon5} we use the asymptotic conditions on the 
background functions
\eqref{derivpot}, that $\omega$, $\psi$ and $\chi$ have compact support, the bounds
\be
|H_3|\leq C_{19},\qquad |H_4|\leq C_{20}
\ee
and combined all the constants into $C_{17}$. Then we have 
\begin{equation}
\|D_2G(\lambda,u_1)[v]-D_2G(\lambda,u_2)[v]\|_{L'^2_{-5/2}}\leq C_6\|v\|_{H'^2_{-1/2}}\|u_1-u_2\|_{H'^2_{-1/2}} 
\end{equation}
proving that the derivative operator \eqref{D2G} is also continuous.

\subsection{The map $D_2G(0,0):H'^2_{-1/2}\ra L'^2_{-1/2}$ is an isomorphism}
Finally, we need to prove that $\mathcal L:=-D_2G(0,0):H'^2_{-1/2}\ra L'^2_{-1/2}$ is an isomorphism. 
As in \cite{Dain:2010uh}, we can write
\begin{equation}\label{oppositive}
 D_2G(0,0)[v]=\Delta v-\alpha v,
\end{equation}
\begin{equation}\label{alpha0}
 \alpha=-\frac{\Delta_2q_0}{4}+7\frac{(\partial\omega_0)^2}{16\rho^4\Phi_0^8}+3\frac{(\partial\psi_0)^2+(\partial\chi_0)^2}{4\rho^2\Phi_0^4},
\end{equation}
which can be written as 
\begin{equation}\label{alphah}
 \alpha=h r^{-2}
\end{equation}
and $h$ is a bounded function in $\mathbb{R}^3$ with $h\in L^2(M)$. In \cite{gabach09} it was proven that when $h$ is positive, the operator \eqref{oppositive} is an isomorphism. 
In general, due to the first term in \eqref{alpha0}, $\alpha$ is not necessarily positive.  
However, here is where the Yamabe positivity condition plays a role. We have the 
following important result.

\begin{lemma}
 Let $(M,\tilde g_{ij})$ be in the positive Yamabe class, namely
\be 
 \int_M|\partial f|_{\tilde g} ^2+\tilde Rf^2 d\mu_g>0
 \ee
 for all $f\in C_0^\infty$, $f\neq0$, then
\be\label{yamflat} 
 \int_M|\partial f|^2+\alpha f^2 d\mu>0
 \ee
where $\alpha$ is given in \eqref{alpha0} and  the norm and volume
element in \eqref{yamflat} are computed with respect to the flat metric.
 
\end{lemma}

\textbf{Proof}
We start with the left hand side of \eqref{yamflat}
 \begin{eqnarray}
 & &\int_M\left[|\partial f|^2+\alpha f^2\right]d\mu=\\
&=& \int_M\left[|\partial f|^2 -\frac{\Delta_2q_0}{4}f^2\right]d\mu+\int_M\left(7\frac{(\partial\omega_0)^2}{16\rho^4\Phi_0^8}+3\frac{(\partial\psi_0)^2+(\partial\chi_0)^2}{4\rho^2\Phi_0^4}\right)f^2d\mu\geq\\ \label{penu}
&\geq&\int_M\left[e^{-2q_0}|\partial f|^2 -2e^{-2q_0}\frac{\Delta_2q_0}{8}f^2\right]e^{2q_0}d\mu=\\\label{yama}
&=&\int_M\left[|\partial f|_{\t{g}}^2 +\tilde R f^2\right]d\mu_{\t{g}}>0,
\end{eqnarray}
which proves the claim. Note that in order to go from \eqref{penu} to \eqref{yama} we have 
used the background metric $\tilde g=e^{2q}(dr^2+r^2 d\theta^2)+r^2\sin^2\theta d\phi^2$.
\begin{flushright}$\Box$\end{flushright}

 \begin{theorem}\label{teoiso}
The linear map $\mathcal L$ defined by
\begin{equation}\label{ecL}
\mathcal Lu:=-\Delta u+\alpha u=f \qquad \mbox{in}\; \mathbb R^{3}\setminus\{0\},
\end{equation}
where $\alpha$ is given by \eqref{alpha0}-\eqref{alphah} and satisfies \eqref{yamflat}, is an isomorphism $H^{'2}_{-1/2}\to L^{'2}_{-5/2}$.
\end{theorem}
The proof of this result will be given below  and departs slightly from \cite{Dain:2010uh} because we exploit the symmetry of the weak problem
associated to \eqref{ecL} to apply the Riesz Representation theorem instead of Lax-Milgram theorem used in \cite{gabach09}. 
This is important as we will no longer need to prove the coercivity condition.

We first prove the existence of a weak solution (Lemma \ref{lemadebil}) and then we find it to be regular.

\begin{lemma} \label{lemadebil}
There exists a unique weak solution $u\in H^{'1}_{-1/2}$ of \eqref{ecL}
for each $f\in L^{'2}_{-5/2}$.
\end{lemma}

\begin{proof}
For $u,v\in H^{'1}_{-1/2}$, define the bilinear form 
\begin{equation}\label{Bil}
B[u,v]:=\int_{\Rt}\partial u\cdot \partial v+\alpha uv\,d\mu
\end{equation}
which corresponds to the linear operator $\mathcal L$.

Let us check that $B[\;,\,]$ satisfies the hypotheses of Riesz Representation theorem (see 
\cite{Evans98}). We first need to prove that the $B[u,v]$ can be taken as an inner product on 
$H^{'1}_{-1/2}\times H^{'1}_{-1/2}$. By the Yamabe condition we know that for all $u\not\equiv0$, $B[u,u]>0$ and also, by definition, if 
$u\equiv0$, then $B[u,u]=0$. Therefore, the bilinear form is positive definite. Second, it can 
easily be proven that $B[u,v]=B[v,u]$ and that $B[u,av+cw]=aB[u,v]+bB[u,w]$. Therefore, $B[u,v]$ is 
an inner product. Next we need to prove that the linear functional $\ell(\cdot):=B[\cdot,v]$ is 
bounded for all $v\in H^{'1}_{-1/2}$. This is done exactly as in \cite{gabach09} 
\begin{eqnarray}
\left|B[u,v]\right|&\leq& \left|\int\partial u\cdot\partial v d\mu\right|+\left|\int \alpha u v d\mu\right|\\
&\leq& |\partial u|_{L^2}|\partial v|_{L^2}+C|ur^{-1}|_{L^2}| ur^{-1}|_{L^2}\\
&\leq& |\partial u|_{L^2}|\partial v|_{L^2}+C|u|_{L'^2_{-1/2}}| u|_{L'^2_{-1/2}}\\
&\leq&\max\{1,C\} |u|_{H'^1_{-1/2}}| u|_{H'^1_{-1/2}}.
\end{eqnarray}

Then with these conditions fulfilled, Riesz Representation Theorem states that there exists a 
unique $u\in H^{'1}_{-1/2}$  such that
\begin{equation}
B[u,v]=\langle f,v\rangle,\qquad \forall v\in H^{'1}_{-1/2}, 
\end{equation}
that is, such that
\begin{equation}
\int_{\Rt}(\mathcal Lu-f)vdx=0,\qquad \forall v\in H^{'1}_{-1/2}.
\end{equation}
Therefore $u$ is the unique weak solution of $\mathcal Lu=f$. 
\end{proof}

Next, we use Lemma A.3 in \cite{gabach09} to prove regularity of solution, namely

\vspace{0.25cm}

\textit{\textbf{Lemma A.3 in \cite{gabach09}.} Let $f\in L^ {'2}_{-5/2}$. Assume $u\in H^ {'1}_{-1/2}$ is a weak solution of $\mathcal Lu=f$. Then $u\in H^ {'2}_{-1/2}$.}

\vspace{0.25cm}

These two lemmas show that there exists a unique function $u\in H^{'2}_{-1/2}$ which solves equation $-\Delta u+\alpha u=f$ a.e, for each $f\in L^{'2}_{-5/2}$. This, in turn, means that $\mathcal L:=-\Delta+\alpha$ is an isomorphism $H^{'2}_{-1/2}\to L^{'2}_{-5/2}$, proving Theorem \ref{teoiso}.

\appendix
\section{Time-rotation symmetry}\label{Aptimerotation}

An axially symmetric initial data $(M,g_{ij},K_{ij}, E^i,B^i)$ has the time-rotation symmetry if, in the coordinates associated with the axial symmetry, under the map $\phi\to-\phi$ the initial data map as
\be
g_{ij}\to g_{ij}, \qquad K_{ij}\to-K_{ij},\qquad
E^i\to E^i,\qquad B^i\to B^i.
\ee
This symmetry on the level of the initial data implies that the development is invariant under the transformation $(t,\phi)\to(-t,-\phi)$ (see \cite{Bardeen70}, \cite{Hawking73b}, \cite{Brandt:1996si}).

Using the symmetries it can be concluded that (see \cite{Chrusciel:2009ki}, \cite{Dain06c})
\begin{equation}
 \t{K}^{ij}\t{K}_{ij} =  \frac{|D\omega|_{\t{g}}^2}{2|\eta|_{\t{g}}^4}=e^{-2q}\frac{(\partial\omega)^2}{2\rho^4}.
\end{equation}

We consider now the electric and magnetic fields. Since the initial data is axially symmetric, the components of the fields in $(\rho,z,\phi)$ coordinates do not depend in the $\phi$ coordinate. This means in particular that
\begin{equation*}
 E^3(\rho,z,-\phi) = E^3(\rho,z,\phi).
\end{equation*}
On the other hand, the discrete symmetry $\phi\ra -\phi$ on the initial data implies
\begin{equation*}
 E^3(\rho,z,-\phi) = -E^3(\rho,z,\phi),
\end{equation*}
and therefore
\begin{equation}
 E^3(\rho,z,\phi) = 0.
\end{equation}
Taking into account that
\begin{equation}
 \eta = d\phi,
\end{equation}
we can write the previous condition in a coordinate invariant way as
\begin{equation}\label{condE}
 E^i\eta_i = 0.
\end{equation}
The condition
\begin{equation}\label{condB}
 B^i\eta_i=0
\end{equation}
is proven in an analogous way. We can reconstruct the fields from the potentials, and using \eqref{condE} and \eqref{condB} we have
\begin{eqnarray}
\label{Eipot} E^i & = & \frac{1}{|\eta|^2}\epsilon^{ijk} \eta_j \partial_k\psi,\\
\label{Bipot} B^i & = & -\frac{1}{|\eta|^2}\epsilon^{ijk} \eta_j \partial_k\chi.
\end{eqnarray}
Rescaling and taking the norm we finally arrive at expressions \eqref{fieldspot}

From the electric and magnetic field we can reconstruct the electromagnetic tensor
\begin{equation}\label{FfromEB}
 F_{ij} = 2 E_{[i}n_{j]} - B^k n^l \epsilon_{klij},
\end{equation}
where $n$ is the normal to the surface $M$. In terms of the potentials
\begin{equation}
 F_{ij} = \frac{1}{|\eta|^2}\left(2\eta_{[i}\partial_{j]}\chi - \eta_{k}\partial_{l}\psi\epsilon^{kl}\,_{ij}\right).
\end{equation}

\section{Maxwell equations in terms of the potentials}\label{ApEinstein}
Here we show that the Maxwell equations
\be
\nabla\cdot\me=0,\qquad \nabla\cdot\mb=0
\ee
are automatically satisfied by any choice of potentials $\psi, \chi$. In terms of the 
potentials we have (see \cite{Chrusciel:2009ki}-\cite{Costa:2009hn})
\be
E_1=\frac{1}{\sqrt{g_{33}}}\partial_2\psi,\qquad E_2=-\frac{1}{\sqrt{g_{33}}}\partial_1\psi,\qquad E_3=\frac{1}{\sqrt{g_{33}}}\partial_n\chi
\ee

In components we write

\begin{eqnarray}
\nabla\cdot\me&=&g^{ab}\nabla_aE_b=g^{ab}(\partial_aE_b-\Gamma_{ab}^cE_c)\\
&=&g^{11}(\partial_1E_1-\Gamma_{11}^1E_1-\Gamma_{11}^2E_2-\Gamma_{11}^3E_3)+\\
&+&g^{22}(\partial_2E_2-\Gamma_{22}^1E_1-\Gamma_{22}^2E_2-\Gamma_{22}^3E_3)+\\
&+&g^{33}(\partial_3E_3-\Gamma_{33}^1E_1-\Gamma_{33}^2E_2-\Gamma_{33}^3E_3)+\\
\end{eqnarray}

Due to axial symmetry we have $\partial_3E_3=0$, $\Gamma_{33}^3=0$  and $g_{11}=g_{22}$, which leaves us with

\begin{eqnarray}
\nabla\cdot\me&=&g^{11}\partial_2\psi\left[\partial_1(g^{-1/2}_{33})- g^{-1/2}_{33}(\Gamma_{11}^1+\Gamma_{22}^1+g_{11}g^{33}\Gamma_{33}^1)\right]-\\
&-&g^{11}\partial_1\psi\left[\partial_2(g^{-1/2}_{33})-g^{-1/2}_{33}(\Gamma_{11}^2+\Gamma_{22}^2+g_{11}g^{33}\Gamma_{33}^2)\right]-\\
&-&g^{11}g^{-1/2}_{33}\partial_n\chi(\Gamma_{11}^3+\Gamma_{22}^3)
\end{eqnarray}
Now use that
\be
-\Gamma_{22}^1=\Gamma_{11}^1=\frac{1}{2}g^{11}g_{11,1}
\ee
\be
-\Gamma_{11}^2=\Gamma_{22}^2=\frac{1}{2}g^{11}g_{11,2}
\ee
\be
\Gamma_{ii}^3=0
\ee

\be
\Gamma_{33}^i=-\frac{1}{2}g^{11}g_{33,i}
\ee
to obtain
\begin{eqnarray}
\nabla\cdot\me&=&g^{11}\partial_2\psi\left[-\frac{1}{2}g^{-3/2}_{33}\partial_1g_{33}+\frac{1}{2}g^{-3/2}_{33}\partial_1g_{33}\right]-\\
&-&g^{11}\partial_1\psi\left[-\frac{1}{2}g^{-3/2}_{33}\partial_2g_{33}+\frac{1}{2}g^{-3/2}_{33}\partial_2g_{33}\right]=\\
&=&0.
\end{eqnarray}
And similarly for $\nabla\cdot\mb=0$. This means that Maxwell constraints are automatically satisfied when the fields are written in terms of the potentials $\psi,\chi$, 
leaving no equations for the potentials.

\end{document}